\documentclass{ifacconf}

\newtheorem{definition}[thm]{Definition}
\newtheorem{lemma}[thm]{Lemma}
\newtheorem{assumption}[thm]{Assumption}
\usepackage{graphicx}      
\usepackage{natbib}        

\usepackage{graphicx}
\usepackage{epstopdf}
\usepackage{xcolor}
\usepackage{algorithm}
\usepackage{algpseudocode}

\newtheorem{remark}[thm]{Remark}
\usepackage{amsmath,amssymb,amsfonts}  

\usepackage{anyfontsize}
\usepackage{booktabs}
\usepackage{graphicx}
\begin{document}
\begin{frontmatter}

\title{Topology Estimation for Open Multi-Agent Systems\thanksref{footnoteinfo}} 

\thanks[footnoteinfo]{This work is supported by the Swedish Research Council (VR) and the	Knut and Alice Wallenberg Foundation.}

\author[First]{Nana Wang},
\author[First]{Pelin Sekercioglu}, 
\author[First]{Dimos V. Dimarogonas}

\address[First]{KTH Royal Institute of Technology, SE-100 44 Stockholm, Sweden (e-mail: \{nanaw, pelinse, dimos\}@kth.se).}

\begin{abstract}                
We address the problem of interaction topology identification in open multi-agent systems (OMAS) with dynamic node sets and fast switching interactions. In such systems, new agents join and interactions change rapidly, resulting in intervals with short dwell time and rendering conventional segment-wise estimation and clustering methods unreliable. To overcome this, we propose a projection-based dissimilarity measure derived from a consistency property of local least-squares operators, enabling robust mode clustering. Aggregating intervals within each cluster yields accurate topology estimates. The proposed framework offers a systematic solution for reconstructing the interaction topology of OMAS subject to fast switching. Finally, we illustrate our theoretical results via numerical
simulations.

\end{abstract}


\end{frontmatter}

\section{Introduction}

Open multi-agent systems (OMAS) are networks, where agents can join or leave and modify their interaction over time. They naturally encompass a broad family of physical and biological systems, including social \citep{hendrickx2016open,franceschelli2020stability}, biological \citep{perelson1997immunology, dimitrov2024liana+}, and engineered systems \citep{restrepo2022consensus}. In several works on OMAS (see \textit{e.g.,} \cite{xue2022stability,restrepo2022consensus,CDC2025OMAS}), openness is captured using the switched system representation, where the system switches modes whenever at least one node or edge is added or removed. The underlying topological structure plays a central role in shaping system-level properties and its collective behaviors. However, in all of these works, the network is assumed to be \textit{known}. This motivates the need for topology estimation methods capable of handling uncertain, time-varying, and structurally evolving networks. 

Relevant insights stem from the literature on identification of switched linear systems, which provides a conceptual framework for analyzing systems characterized by multiple (interaction) modes. In this context, a mode represents a communication graph, and the switch occurs among different modes.
The works in \citep{massucci2021regularized} and \citep{massucci2022statistical} investigated switched discrete-time linear systems, leveraging statistical learning theory to derive identification guarantees. In their approach, the system matrix is first estimated for each time segment (when one mode dwells consecutively until the next mode is executed) and the resulting matrices are clustered based on their pairwise distances. Then, the system matrix for each mode is re-estimated using all time segments belonging to the same mode, and this procedure is repeated until convergence.  However, since short time segments contain limited information, estimating the system matrix directly from each time segment is inaccurate. Hence, if two time segments share the same system matrix, their limited data can produce substantially different least-squares estimates, leading to incorrect estimation and mode assignments, and causing poor performance.  Several works \citep{Wang2025auto, sun2023identifiability} have further examined the identifiability and recovery of switching network topologies. These works established conditions under which the topology in each mode can be uniquely identified. \citep{rey2025online} studied the online estimation problem for expanding graphs via regularized least-squares methods promoting temporal smoothness. Although these contributions significantly advance the understanding of time-varying and switching networks, they generally assume fixed or monotonically expanding node sets, or rely on slowly varying topology changes. Such assumptions limit their applicability in OMAS characterized by rapid agent arrivals, departures, and abrupt interaction changes. 

Motivated by the need to develop new topology estimation methods for OMAS, we simultaneously consider dynamic node-set variations and fast switching interactions.   A key challenge in this setting is that each short time segment  contains limited information, rendering direct segment-wise topology estimation unreliable. Thus, our contributions are threefold.  First, we propose algorithms that estimate each mode by aggregating information across all time segments in which the same mode occurs, in contrast to \citep{sun2023identifiability}, which  relies solely on one time segment estimates obtained. Second, we introduce a new distance criterion, called projection-based dissimilarity measure, that compares time segments without requiring accurate topology estimates for each time segment. This criterion enables robust clustering of time segments belonging to the same mode, even under sparse observations and rapidly changing agent sets. Third, building on these ideas, we develop a two-stage topology reconstruction framework that first clusters time segments according to their latent modes and then aggregates information within each cluster to estimate the topology associated with each mode. This framework provides a resilient solution for reconstructing the interaction topologies in OMAS with dynamic and evolving population.

\begin{figure*}[!]  
    \centering
    \includegraphics[width=\textwidth]{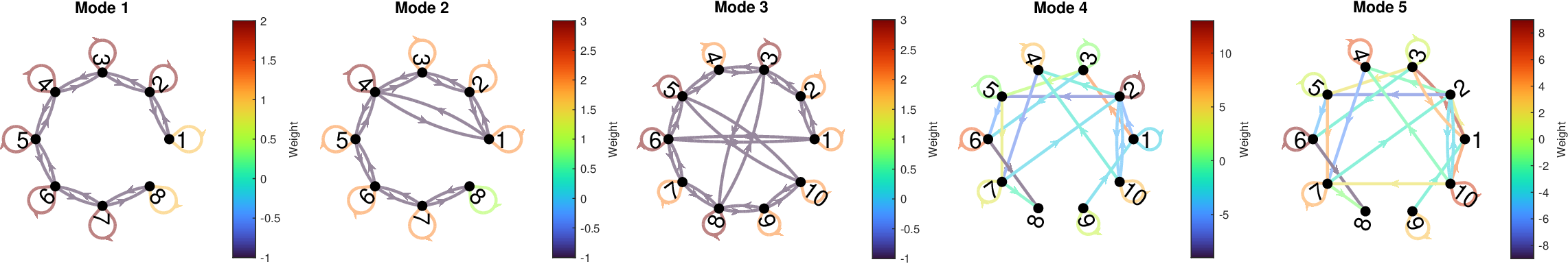}
    \caption{Example of directed weighted graphs representing each $L^{\mathcal{T}(k)}$. The color of the edges denotes the corresponding weight, as indicated by the color bar for each graph. Modes $1$ and $2$ correspond to the modes for an $8$-agent network, while Modes $3-5$ correspond to the  modes for an $10$-agent network.}
    \label{fig:directedweight}
\end{figure*}

The structure of this paper is as follows: the problem formulation is presented in
Section \ref{Sec:problemf}.   In Section \ref{sec:frame}, we introduce two algorithms to estimate the switching topology for each interval and combine the information collected from multiple intervals to estimate its connectivity matrices.  Section \ref{sec:clustering} presents an algorithm to cluster time segments with the same mode and a framework to combine mode clustering and topology estimation.  
Section \ref{Sec:simu} verifies the proposed algorithms' effectiveness with a numerical example and Section \ref{Sec:con} concludes the paper.

\textit{Notations.} $\mathbb{R}^{n\times m}$ and $\mathbb{R}_{\ge 0}$ denote the set of the real matrices with approximate dimensions,  and non-negative real numbers.  $\mathbb{N}^+$ denotes the set of positive nature numbers. 
$\lfloor x \rfloor$
denotes the floor of $x\in \mathbb{R}$, i.e., the largest integer less than or equal to $x$. $X\succ 0$ denotes that  $X$ is a positive definite matrix. For matrices $X$ and $Y$, $X \succ Y$ means that $X-Y \succ 0$. For a matrix $A \in \mathbb{R}^{n\times m}$, we denote its null space and range by $\ker (A)$ and $\text{ran}(A)$ respectively, \textit{i.e.,} $\ker (A)=\{x \in \mathbb{R}^m | Ax=0\}$ and $\text{ran}(A)=\{y \in \mathbb{R}^n | y=Ax,  x \in \mathbb{R}^n\}$. $A^\top$ denotes the transpose of $A$.  $\|\cdot\|_F$ denotes the Frobenius norm. For a set $S$, $|S|$ denotes its cardinality. For a matrix $A$, $A^T $ denotes its transpose and vec(A) is the vectorization operator that stacks the columns of the matrix $A$ into a single column vector.

\section{Problem formulation}
\label{Sec:problemf}

Let $\mathcal{G}=\big(\mathcal{V},\mathcal{E},W\big)$ be a directed and weighted graph representing a mode, where $\mathcal{V}=\{1,2,\dots,N\}$ denotes the set of nodes corresponding to the agents, $\mathcal{E}\subseteq \mathcal{V}\times\mathcal{V}$ is the set of edges interconnecting the nodes, and $W$ contains the associated edge weights. We consider a multi-agent system, interconnected over a directed and weighted graph, with dynamic interactions, where agents can join and interactions between agents can change over time.
 The time-varying interactions are represented by a switched system representation, where each mode corresponds to a fixed topology and a switch occurs whenever the interaction changes.
The switching instants 
\[
    0=t^0 < t^1 < \cdots < t^\kappa < t^{\kappa+1}=+\infty
\] correspond to events where (i) a new agent joins, (ii) an existing agent leaves, or (iii) the interaction among agents changes (\textit{i.e.,} a new edge is created between agents or the edge weight changes). The dwell time of the $k$-th interval is defined as
$\tau^k = t^{k}-t^{k-1},$ where $k = 1,\dots,\kappa+1$ and $\kappa\in \mathbb{N}^+$.

We denote the mode set by $\mathcal{P} = \{1, \dots, M \}$, where each mode corresponds to a distinct interaction topology of the network. The value $M$ represents the total number of modes that appear in the OMAS (See Remark \ref{re:infinity}, which discusses the potential infiniteness of $M$ and $\kappa$). We define the mapping 
\[
    \mathcal{T}: \{0,1,\dots,\kappa\} \rightarrow \mathcal{P}.
\]

Now consider $k =  0,1,\cdots, \kappa$. During the interval $[t^{k},t^{k+1})$,  the system operates in mode $\mathcal{T}(k)$; we say that this interval belongs to mode $\mathcal{T}(k)$. For $t\in [ t^k, t^{k+1} )$, the system evolves according to
\begin{equation}\label{sys}
    \dot{x}(t)
    = f(x(t)) + L^{\mathcal{T}(k)} G(x(t)) + u(t),
\end{equation}
where $x(t)\in\mathbb{R}^{N_{\mathcal{T}(k)}}$ denotes the collective state of the agents in the active node set $\mathcal{V}_{\mathcal{T}(k)}$, and $N_{\mathcal{T}(k)}=|\mathcal{V}_{\mathcal{T}(k)}|$. The function $f(x(t))$ represents the known internal dynamics, with $f(x(t))= [f_1(x_1(t)), \cdots, f_{N_{\mathcal{T}(k)}}(x_{N_{\mathcal{T}(k)}}(t))]^T$. The matrix $L^{\mathcal{T}(k)}\in\mathbb{R}^{N_{\mathcal{T}(k)}\times N_{\mathcal{T}(k)}}$ is the connectivity matrix representing the unknown weighted interaction topology associated with mode ${\mathcal{T}(k)}$, $G(x(t))=x(t)$  specifies the interaction protocol (e.g., consensus or formation control \citep{mesbahi}), and $u(t)\in\mathbb{R}^{N_{\mathcal{T}(k)}}$ denotes the externally applied control input. We emphasize that both the dimension $N_{\mathcal{T}(k)}$ and the connectivity matrix $L^{\mathcal{T}(k)}$ depend on the mode ${\mathcal{T}(k)}$ ---see Figure  \ref{fig:directedweight} for an example of different modes with different networks sizes. If $L^{\mathcal{T}(k)}_{il} \ne 0$, then there is a directed edge $e_k \in \mathcal{E}_{\mathcal{T}(k)}$ from node $l$ to node $i$ with weight $L^{\mathcal{T}(k)}_{il}$ for all $i,l \in \{1,\dots, N_{\mathcal{T}(k)}\}$ with $i\ne l$.  If $L^{\mathcal{T}(k)}_{ii} \ne 0$, then node $i$ has a self-loop, which belongs to $\mathcal{E}_{\mathcal{T}(k)}$. Thus,  $\mathcal{V}_{\mathcal{T}(k)}$ consists of all active nodes of the system \eqref{sys}, while the edge set $\mathcal{E}_{\mathcal{T}(k)}$ and the edge weight set $W_{\mathcal{T}(k)}$ are derived from $L^{\mathcal{T}(k)}$. We denote the communication graph of \eqref{sys} as a directed weighted graph $\mathcal{G}_{\mathcal{T}(k)}=(\mathcal{V}_{\mathcal{T}(k)},\mathcal{E}_{\mathcal{T}(k)},W_{\mathcal{T}(k)})$, where $\mathcal{V}_{\mathcal{T}(k)}$ is the node set, $\mathcal{E}_{\mathcal{T}(k)}\subseteq \mathcal{V}_{\mathcal{T}(k)}\times\mathcal{V}_{\mathcal{T}(k)}$ is the edge set, and $\mathcal{W}_{\mathcal{T}(k)}$ denotes the associated edge weights.
We are now ready to present our objective.

\noindent\textbf{Objective.}
Given the observed output $y(t)$, the applied input $u(t)$, and the switching sequence $\{t^k\}_{k=0}^{\kappa}$, the goal is to reconstruct the sequence of underlying topologies:
\[
    \left\{ L^{\mathcal{T}(k)} \right\}_{k=0}^{\kappa}.
\]

\section{Estimation over Single and Multiple Intervals}
\label{sec:frame}
In this section, we present a method to excite the network through the external control input $u$ and then estimate the interaction topology. We first propose an algorithm (see Algorithm \ref{alg:algorithm1}) that estimates the interaction topology by comparing an auxiliary system $\hat x$ with the original network $x$, and by estimating the connectivity matrix over each interval. Then, we propose a more flexible method (see Algorithm \ref{alg:algorithm2}) that aggregates information across multiple intervals to estimate the connectivity matrix. Throughout this section, we assume that the mode labels are known, an assumption that will be relaxed in Section \ref{sec:clustering}.

 \subsection{Network Control Design and Topology Estimation for Every Interval}
\label{sub:one mode}
 
To estimate the topology of a mode that may appear across multiple   intervals with short dwell time, we begin by recalling classical excitation definitions for system identification. We adopt the notions of persistence of excitation and sufficient richness, both standard in adaptive control \citep{ioannou1996robust}.
\begin{definition}\label{PE}(PE condition  \citep{ioannou1996robust})
A signal $\phi:\mathbb{R}_{\geq 0}\to \mathbb{R}^{n\times m}$ is said to be persistently exciting (PE)  if, for any $t\ge 0$, there exists  $\gamma>0$ and $T>t$ such that
	\begin{equation}\label{eq:E}
\int_{t}^{t+T} \phi(\tau) \phi(\tau)^\top d \tau \succ \gamma I.
	\end{equation}
  \end{definition}

 \begin{definition}\label{d1e}(Sufficiently rich of order $n$ \citep{ioannou1996robust})
A function $\phi:\mathbb{R}_{\geq 0}\to \mathbb{R}^{n\times m}$ is said to be sufficiently rich of order $n$ if it consists of at least of $n/2$ distinct frequencies.
  \end{definition}
  
To actively inject sufficient excitation into the network \eqref{sys}, we design the control input $u=[u_1, \cdots, u_{N_{\mathcal{T}(k)}}]^T$ as
\begin{equation}
\label{control}
    u_i(t)=\hat u_i(t) -f_i(x_i(t)),
\end{equation}
where $\hat u(t)=[\hat u_1, \cdots, \hat u_{N_{\mathcal{T}(k)}}]^T$ is designed to be sufficiently rich of order $N_{\mathcal{T}(k)}$, satisfying Definition \ref{d1e}. To this end, we propose Algorithm~\ref{alg:algorithm1}, where we reconstruct the connectivity matrix $L^{\mathcal{T}(k)}$ on each interval $[t^{k},t^{k+1})$ by implementing an auxiliary system $\hat x(t)$, described below in \eqref{eq:RK_systema}, provided that the output $w(t)$ of the filter in \eqref{eq:RK_systemb} satisfies an excitation condition. To write the overall dynamics compactly, we define the augmented state $s(t) = \big(\hat x(t),\, w(t),\, \mathrm{vec}(Y(t)),\, \mathrm{vec}(Z(t))\big)$, which collects the auxiliary state, the filter state, and the integral quantities involved in the estimation. Its dynamics is given by
\begin{subequations}\label{eq:RK_system}
\begin{align}
    \dot{\hat x}(t) &= f(x(t)) + L_m x(t) + u(t) + \tau(x(t) - \hat x(t)), \label{eq:RK_systema}\\
    \dot w(t) &= x(t) - \tau w(t),\label{eq:RK_systemb}\\
    \dot Y(t) &= w(t) w(t)^\top,\label{eq:RK_systemc}\\
    \dot Z(t) &= (L_m w(t) + x(t) - \hat x(t) - \zeta(t)) w(t)^\top,\label{eq:RK_systemd}
\end{align}
\end{subequations}
where $\tau>0 $ and $L_m \in \mathbb{R}^{N_{{\mathcal{T}(k)}}\times N_{{\mathcal{T}(k)}}}$  represents the connectivity matrix of the auxiliary system $\hat x(t)$ on each interval and can be initialized, \textit{e.g.,} as the zero or the identity matrix.  The step size $h$ in Algorithm~\ref{alg:algorithm1} is chosen to be small to ensure higher accuracy of system \eqref{eq:RK_system}. 

\begin{algorithm}[t]
\caption{Topology Identification Algorithm over Single Interval}
\label{alg:algorithm1}
\begin{algorithmic}[1] 
\State Initialize the parameters $\tau>0$ and $\gamma>0$; the switching time sequence $\{t^k\}_{k=0}^{\kappa}$ with initial and final times $t_0$ and $t_{\kappa}$; and the step size $h>0$.
\For{$k = 0,1,\dots, \kappa$}
\State Initialize the matrix $L_m \in \mathbb{R}^{N_{{\mathcal{T}(k)}}\times N_{{\mathcal{T}(k)}}}$; $\hat x(t^k) = w(t^k) = 0; \
        Y(t^k) = Z(t^k) = 0;\
        \tilde x(t^k) = x(t^k).$
    
    \For{$\ell = 0,1,2,\dots$ such that $t_\ell = t^k + \ell h < t^{k+1}$}
        \State Design $u(t_\ell) \in \mathbb{R}^{N_{{\mathcal{T}(k)}}}$ using \eqref{control} and update the state $x(t_\ell)\in \mathbb{R}^{N_{{\mathcal{T}(k)}}}$ according to the dynamics \eqref{sys}.
        \State Compute $\zeta(t_\ell) = e^{-\tau (t_\ell - t^k)} \tilde x(t^k)$, where $\tilde x(t^k):= x(t^k) - \hat x(t^k)$.
        \State Update the augmented state $s(t_\ell)$ according to the dynamics \eqref{eq:RK_system}.
\State Perform one Runge--Kutta\footnotemark\ step with step size $h$ on \eqref{eq:RK_system}:
\[
    s(t_{\ell+1}) = \mathrm{RK}(s(t_\ell), t_\ell, h).
\]
       
        \State Extract $(\hat x(t_{\ell+1}), w(t_{\ell+1}), Y(t_{\ell+1}), Z(t_{\ell+1})$ from $s(t_{\ell+1})$.
        \State Update time: $t_{\ell+1} = t_\ell + h$.
        \If{$Y(t_{\ell+1}) \succ \gamma I$}
            \State Calculate the estimated topology matrix 
\begin{equation}\label{estimation_A}
                \hat L(t_{\ell+1}) = Z(t_{\ell+1})\, Y(t_{\ell+1})^{-1}.
            \end{equation}
            \State \textbf{Go to Line 2}.
        \EndIf
\EndFor

  \EndFor
\end{algorithmic}
\end{algorithm}
\footnotetext{A standard explicit numerical integration method for solving ordinary differential equations (See \citep{hairer1993solving}).}

\begin{remark}
\label{re:infinity}
For Algorithms~\ref{alg:algorithm1} and \ref{alg:algorithm2}, the switching times $\kappa$ and the number of modes $M$ can be infinite since the model label for each interval is assumed to be known. In contrast, Algorithm~\ref{alg:mode_clustering}, which determines the mode labels (see Section \ref{sec:clustering}), requires $\kappa$ and $M$ to be finite in order to guarantee that mode matching across time segments is trackable. 
\end{remark}
In the following, we provide theoretical guarantees for Algorithm~\ref{alg:algorithm1}, ensuring that the OMAS maintains certain excitation and allows accurate topology estimation under a dwell time assumption (Assumption \ref{assump:2}). First, we recall the following result from \citep[Lemma 4.8.3]{ioannou1996robust}, which states that persistence of excitation is preserved through a stable, minimum-phase filter, which will be used in our results.
\begin{lemma}
    \label{lemma1}
   If $x(t)$ is bounded and PE, then $w(t)$, with the update law \eqref{eq:RK_system}, is also PE. 
\end{lemma}
\begin{pf}
From Algorithm \ref{alg:algorithm1}, the filter $\dot w=x-\tau w$ implies that its transfer function
$H(s)=\frac{1}{s+\tau}$
 is stable and minimum phase. Since $x$ and $\dot x$ are bounded, $H(s)$ satisfies the conditions  in \citep[Lemma 4.8.3]{ioannou1996robust}, and thus PE is preserved through the filter. 
\end{pf}
 We now pose the following assumptions for system \eqref{sys}.
\begin{assumption} \label{assump:1}
Let $L^{\mathcal{T}(n)}$ be Hurwitz on a single interval $[t^n, t^{n+1})$, where $n \in \{0,1,\cdots, \kappa \}.$ 
\end{assumption}
\begin{assumption} \label{assump:2}
The dwell time of the interval $[t^n, t^{n+1})$ for mode $\mathcal{T}(n)$, where $n \in \{0,1,\cdots, \kappa \}$, is sufficiently long to ensure that the filtered signal $w$ is PE. 
\end{assumption}
\begin{lemma}
\label{lemma2}
 Assume that the network states $x(t)$ in \eqref{sys} are bounded. Under  Assumptions \ref{assump:1} and \ref{assump:2}, for the network modeled by the dynamics \eqref{sys}, the control input \eqref{control} ensures that $w$ is PE for mode $\mathcal{T}(n)$ over the interval $[t^n, t^{n+1})$. 
\end{lemma}
\begin{pf} 
    With $u(t)=\hat u(t) -f(x(t))$, the dynamics of \eqref{sys} can be written as $\dot x(t)= L^{\mathcal{T}(n)} x(t) + \hat u(t)$.  $\hat u$ is chosen to be sufficiently rich of order $N_{\mathcal{T}(k)}$ and bounded. The network is fully controlled from the form of \eqref{sys} and under  Assumption \ref{assump:1}, $L^{\mathcal{T}(n)}$ is Hurwitz. Moreover, Assumption \ref{assump:1} ensures that the dwell time is long enough to render $w$ PE.  Then, from the proof of \cite[Theorem 5.2.3]{ioannou1996robust}, $x$ will be sufficiently rich of order $N_{\mathcal{T}(k)}$, thereby $w(t)$ also satisfies PE condition. 
 \end{pf}
The following theorem provides an original contribution of this paper regarding the identification of $L^{\mathcal{T}(n)}$. 
\begin{thm}
\label{theo:1} Consider the network dynamics
\eqref{sys} evolving under a constant interaction topology defined by $L^{\mathcal{T}(k)}$ on the interval $[t^k,t^{k+1})$.
Moreover, the auxiliary system \eqref{eq:RK_system} is simulated by implementing Algorithm \ref{alg:algorithm1}.
Then, for $t\in[t^k,t^{k+1})$, the matrices $Y(t)$ and $Z(t)$, defined in \eqref{eq:RK_system}, satisfy
\begin{equation}
    Z(t) = L^{\mathcal{T}(k)} \, Y(t),\quad k=1,\dots, \kappa.
    \label{eq:ZY_relation_aux}
\end{equation}
Furthermore, under Assumptions \ref{assump:1} and \ref{assump:2}, for $[t^n, t^{n+1})$, $Y(t)$ is invertible for a certain time $t$ satisfying $t\ge t_e$, where $t_e =\{t_{n+1}>t>t^n| Y(t)> \gamma I\}$, thereby the estimator \eqref{estimation_A}
recovers the true topology of mode $n$, i.e.\ $\hat L^{\mathcal{T}(n)}(t) = L^{\mathcal{T}(n)}$, where $\hat L^{\mathcal{T}(n)}(t) $ is the estimate of $L^{\mathcal{T}(n)}$.
\end{thm}
\begin{pf}
 During $[t^k, t^{k+1})$,  since $L_m$ and $\tau$ given in \eqref{eq:RK_system} are fixed, and $u(t)$ is known, we can show that
$\big(L_m w(\sigma) + \tilde x(\sigma) - \zeta(\sigma)\big)$ is equal to
$L^{\mathcal{T}(k)} w(\sigma)$ for all $\sigma\in[t^k,t^{k+1})$ when \eqref{sys}
holds with a constant $L^{\mathcal{T}(k)}$. 
We use $\zeta(\sigma)= \tilde x(\sigma)+\tilde L^{\mathcal{T}(k)} w(\sigma)$ with $\tilde L^{\mathcal{T}(k)}=L_m-L^{\mathcal{T}(k)}$. We obtain  $\dot \zeta(\sigma)=-\tau \zeta(\sigma)$ from the dynamics of \eqref{sys} and \eqref{eq:RK_system}. It corresponds to $ \zeta(t)   =e^{-\tau} \tilde x(t)$. 
Therefore
\begin{align*}
Z(t)
&= \int_{t^k}^{t} L^{\mathcal{T}(k)} w(\sigma) w(\sigma)^\top d\sigma
 = L^{\mathcal{T}(k)} \int_{t^k}^{t} w(\sigma) w(\sigma)^\top d\sigma \\
&= L^{\mathcal{T}(k)}\, Y(t),
\end{align*}
which proves \eqref{eq:ZY_relation_aux}. This holds for all $k=0, 1,\dots, \kappa$. From Lemma \ref{lemma2}, we can obtain that $w$ is persistently exciting during $[t^n, t^{n+1})$. Thereby, $Y(t)> \gamma I$ for $t>t_e$. It also means that $Y(t)$ is invertible for $t>t_e$, the estimator \eqref{estimation_A}
recovers the true mode.\end{pf}

\begin{remark}  The set-up of the network \eqref{sys} requires controlling all nodes. This requirement can be relaxed using the framework presented in \citep{wang2025leader}, which selects a small subset of nodes and applies the control input only to these nodes to excite the network. 
\end{remark}

\subsection{Estimation across Multiple Intervals for a Single Mode}
\label{sub:mone mode}
From Theorem~\ref{theo:1},  accurate identification of one mode requires that it possesses at least one interval which has a sufficiently long dwell time to ensure persistent excitation of the filtered state. However, in many OMAS, certain modes may reappear frequently but only for short durations, making such a requirement unrealistic. To overcome this limitation, in this section, we adopt a more flexible condition based on all occurrences of a mode.

  Formally, let $\mathcal {P}_j=\{k : \mathcal{T}(k)=j\}$ denote the index set of all the intervals during which mode $j$ is active. Algorithm~\ref{alg:algorithm2} is proposed to estimate the connectivity matrix $L^j$ for each mode $j$, which combines different intervals for each mode $j$. 
For $k \in \{0,\dots,\kappa\}$, denote $
\mathcal{Y}_k 
:= \bigl\{\, Y(t_\ell)\; \big|\; \ell \in \mathbb{N},\; t^{k} + \ell h < t^{k+1} \,\bigr\}$ and  $\mathcal{Z}_k 
:= \bigl\{\, Z(t_\ell)\; \big|\; \ell \in \mathbb{N},\; t^{k} + \ell h < t^{k+1} \,\bigr\}.$ Denote  $t_{\mathcal{P}_j(i)}=t^{\mathcal{P}_j(i)-1}+\Bigl\lfloor \frac{t^{\mathcal{P}_j(i)}-t^{\mathcal{P}_j(i)-1}}{h} \Bigr\rfloor h$ if $\mathcal{P}_j(i)\ne 0$ and $t_{\mathcal{P}_j(i)}=0$ if $\mathcal{P}_j(i)= 0$, where  $i=1,2,\cdots, |{\mathcal{P}_j}|$. 


\begin{algorithm}
\caption{ Topology Identification Algorithm over Multiple Intervals}\label{alg:algorithm2}
\begin{algorithmic}[1] 
\State \textbf{Input:} $\{\mathcal{Y}_k\}_{k \in \{0,\dots,\kappa\}}$, $\{\mathcal{Z}_k\}_{k \in \{0,\dots,\kappa\}}$ and step size $h$ from Algorithm~\ref{alg:algorithm1}.
\State Initialize $Y_1(t_0)=Y(t_0)$ and  $Z_1(t_0)=Y(t_0)$.

\For{$i = 2,3,\dots,|\mathcal{P}_j|$}
    \For{$\ell = 0,1,2,\dots,$ such that $t_\ell = t^{\mathcal{P}_j(i)} + \ell h < t^{\mathcal{P}_j(i)+1}$}
        \State Update the accumulated matrices:
        \[
        \begin{aligned}
        Y_1(t_\ell) &= Y_1\!\left(t_{\mathcal{P}_j(i-1)} \right) + Y(t_\ell), \\
        Z_1(t_\ell) &= Z_1\!\left(t_{\mathcal{P}_j(i-1)} \right) + Z(t_\ell).
        \end{aligned}
        \]

        \If{$Y_1(t_\ell) \succ \gamma I$,} \State Compute
            $\hat L^{j} = Z_1(t_\ell)\, Y_1(t_\ell)^{-1}.$
            \State \textbf{Terminate the algorithm and return} $\hat L^{j}$.
        \EndIf
    \EndFor
\EndFor
\end{algorithmic}
\end{algorithm}

Instead of relying on the PE condition within a single interval, we aggregate the filtered data from all intervals indexed by $\mathcal {P}_j$. If the cumulative excitation across these intervals is sufficient, then the topology corresponding to mode $j$ can be uniquely recovered.  We define the set of total intervals of mode \( j \) as $\{ [t^{\mathcal{P}_j(i)},\, t^{\mathcal{P}_j(i)+1} ) \}$  for all $i = 1,\cdots, |\mathcal{P}_j|$, that is, the union of all the intervals in which mode \( j \) is active.    The corresponding dwell time is defined as $\sum_{k\in \mathcal{P}_j} \tau^k $. 

\begin{assumption}\label{assump:3}
The dwell time of the total intervals associated with each mode $j$ is sufficiently long to ensure that the aggregated matrix $Y_2:=
 \sum_{i \in \mathcal {P}_j} Y(t_{\mathcal{P}_j(i)})$ satisfies  $Y_2 \succ \gamma I $, where $I$ is the identity matrix with approximate dimensions.
\end{assumption}

\begin{remark}
Assumption~\ref{assump:1} states that there exists a single interval of a mode which can last sufficiently long for topology estimation, similarly to \citep{sun2023identifiability}. 
Assumption~\ref{assump:2} is significantly less restrictive than Assumption~\ref{assump:1}, as it does not require any single interval of a mode to be long, but only that the \emph{sum} of all the intervals in which the mode is active provides adequate excitation.  It is therefore well suited to settings with fast switching or intermittent agent participation.   
\end{remark}

We now present our result on aggregating all occurrences of a mode to estimate its connectivity matrix. 
\begin{thm}\label{theore2}
Assume that the network states $x(t)$ in \eqref{sys} remain bounded. Under  Assumptions \ref{assump:2} and \ref{assump:3}, for the network modeled by the dynamics \eqref{sys} the control input \eqref{control}, implementing  Algorithms \ref{alg:algorithm1} and \ref{alg:algorithm2} guarantees that the estimation errors $\tilde L^j =\hat L^j -L^j$ are zero for each mode $j$.  
\end{thm}
\begin{pf}
From Theorem~\ref{theo:1}, for each interval $[t^k, t^{k+1})$ satisfying $\mathcal{T}(k)=j$, we have
$Z(t) = L^{j} Y(t)$.
Summing over all the intervals indexed by $\mathcal{P}_j$, we obtain
\[
\sum_{i \in \mathcal {P}_j} Z(t_{\mathcal{P}_j(i)}) 
= 
L^{j} \sum_{i \in \mathcal {P}_j} Y(t_{\mathcal{P}_j(i)}).
\]
These correspond to $Y_1(t)$ and $Z_1(t)$ in Algorithm \ref{alg:algorithm2}. Assumption~\ref{assump:2} ensures that the aggregated matrix $ Y_2
$
is positive definite, and therefore invertible.   The topology for mode $j$ is thus recovered via
\[
\hat L^{j}
=
\bigg( \sum_{i \in \mathcal {P}_j} Z(t_{\mathcal{P}_j(i)}) \bigg)
\bigg( \sum_{i \in \mathcal {P}_j} Y(t_{\mathcal{P}_j(i)}) \bigg)^{-1},
\]
thereby $\tilde L^{j} = 0^{N_j \times N_j}$.\end{pf}

\begin{remark}
If agent arrivals and departures cease after a finite time, Algorithm~\ref{alg:algorithm2} can be applied directly. In contrast, when structural changes continue indefinitely, the method can be executed online: newly appearing vertex sets are treated as new modes, and their topologies are estimated once sufficient dwell time of the total interval has been accumulated.
\end{remark}

\section{A framework for topology estimation in general OMAS}
\label{sec:clustering}
In this section, we develop a method to determine the mode label for each time segment, which is the basis for identifying switching interaction topologies in OMAS. The key idea is to first provide a projected-based  measure between time segments and then cluster time segments that correspond to the same mode. Afterwards, all data associated with each cluster are aggregated to estimate the topology of that mode.

\subsection{Time sequence clustering}
 We assume the switching instants are observable—either directly from agent arrival/exit events or through change detection mechanisms applied to the data. Our goal is therefore to group time segments that share the same underlying mode.

A straightforward idea is to estimate the connectivity matrix for each time segment and cluster the matrices using the difference among the estimated connectivity matrices, as done in existing statistical-learning approaches for switched system identification \citep{massucci2022statistical}. In particular, \citep{massucci2022statistical} uses an expectation–maximization–like procedure that iterates between clustering and parameter estimation. However, this pipeline perform poorly in OMAS. Each short time segment contains limited information, and thus directly estimating the connectivity matrix $L^{j}$ from each time segment is inaccurate. Even if two time segments share the same matrix $L^{j}$, their limited data can produce substantially different least-squares estimates $\hat L^{j}$, resulting in incorrect mode assignments.

To address this challenge, we provide a projection-based distance to measure the distance among the rough topology estimations for each time segment. Let $I_s = [t^s,t^{s+1})$  and $I_r=[t^r,t^{r+1})$ be two intervals, where $s,r \in \{0, 1, \cdots, \kappa\}$ and $s\ne r$.  Denote $t_s=t^{s}+\Bigl\lfloor \frac{t^{s+1}-t^{s}}{h} \Bigr\rfloor h$ and $t_r=t^{r}+\Bigl\lfloor \frac{t^{r+1}-t^{r}}{h} \Bigr\rfloor h$.   Denote the corresponding terminal matrices from Algroithm~\ref{alg:algorithm1} by
\begin{equation}
    Y_s := Y(t_{s} ), \qquad Z_s := Z(t_{s}),
\end{equation}
\begin{equation}
    Y_r := Y(t_{r}), \qquad Z_r := Z(t_{r}),
\end{equation}
and define the corresponding least-squares estimates and projection matrices as
\begin{equation}
    \hat L^s := Z_s Y_s^\dagger, \qquad
    \hat L^r := Z_r Y_r^\dagger, 
    \end{equation}
\begin{equation}
    P_s := Y_s Y_s^\dagger, \qquad
    P_r := Y_r Y_r^\dagger,
\end{equation}
where $(\cdot)^\dagger$ denotes the Moore--Penrose pseudoinverse, and  $P_s$ and  $P_s$  are the orthogonal projection onto $\operatorname{range}(Y_s)$ and $\operatorname{range}(Y_r)$, correspondingly.

\begin{lemma}
\label{lem:mode_consistency}
Assume that on both intervals $I_s$ and $I_r$ the system is in the same connectivity matrix $L$, i.e.,
\[
 Z_s = L Y_s, \qquad  Z_r = L Y_r.
\]
Then the local least-squares operators satisfy
\begin{equation}
\hat L^s P_r = \hat L^r P_r
\quad\text{and}\quad
\hat L^r P_s = \hat L^s P_s.
    \label{eq:projection_identity_Ax_plus_u}
\end{equation}
\end{lemma}
\begin{pf}
From $ Z_s = L Y_s$ we obtain
\[
\hat L^s P_r
=  Z_s Y_s^\dagger P_r
= L Y_s Y_s^\dagger P_r
= L P_s P_r.
\]
Similarly, from $\tilde Z_r = L Y_r$ we have
\[
\hat L^r P_r
=Z_r Y_r^\dagger P_r
= L Y_r Y_r^\dagger
= L P_r.
\]
Since $P_s P_r = P_r$ on $\operatorname{ran}(Y_r)$, it follows that
$\hat L^s P_r = \hat L^r P_r$. The second identity
$\hat L^r P_s = \hat L^s P_s$ is obtained by symmetry, interchanging $s$ and $r$.
\end{pf}
\begin{remark}
Lemma \ref{lem:mode_consistency} implies that $\hat L^s$ and $\hat L^r$ induce the same linear map on each other's excitation subspaces.
\end{remark}

Using Lemma~\ref{lem:mode_consistency}, we define the following projection-based dissimilarity measure among two time-segments:
\begin{equation}\label{d_sr}
    d(s,r)
    :=
    \big\|\hat L^s P_r - \hat L^r P_r\big\|_F
    +
    \big\|\hat L^r P_s - \hat L^s P_s\big\|_F.
\end{equation}
If two intervals share the same mode and are sufficiently excited,
then $d(s,r)$ is small; otherwise, the projection identities fail and $d(s,r)$
is large. Clustering time segments based on $d(s,r)$ therefore yields groups
corresponding approximately to distinct modes.  We only compare the distances among time segments when the vertex set is the same.
Denote by $p \in \mathbb{N}^+$ the number of the vertex sets among all time segments. Hence we propose Algorithm \ref{alg:mode_clustering} for time segment clustering. 

\begin{algorithm}[t]
\caption{Mode Clustering via Projection-Based Dissimilarity}
\label{alg:mode_clustering}
\begin{algorithmic}[1]
\Require Time partition $\{[t^s,t^{s+1})\}_{s=0}^{\kappa}$, terminal matrices $\{Y_s,Z_s\}_{s=0}^{\kappa}$ obtained from Algorithm \ref{alg:algorithm1}, number of modes $M^i$ for each vertex set $\mathcal{V}_i$, $i=1,\dots, p$.  
\Ensure Mode index $\ell_s \in \{1,\dots, M\}$ for each time segment $s=0,\dots,\kappa$.
\vspace{0.5ex}
\Statex \textbf{Step 1: Local operators and projections}
\For{$s = 0,\dots,\kappa$}
    \State Compute local LS estimate
    \[
        \hat L^s := Z_s Y_s^\dagger
    \]
    where $Y_s^\dagger$ is the Moore--Penrose pseudoinverse.
    \State Compute excitation projector
    \[
        P_s := Y_s Y_s^\dagger.
    \]
\EndFor
\vspace{0.5ex}
\Statex \textbf{Step 2: Pairwise dissimilarity matrix}
\State Classify time segments according to the size of vertex set as $\mathcal{V}_1, \dots, \mathcal{V}_p$. 
\State Initialize $D^i \in \mathbb{R}^{|\mathcal{V}_i|\times |\mathcal{V}_i|}$ for each $\mathcal{V}_i$  with $D^i_{s,s} = 0$. 
\For{$i =1,\dots, p $}
\For{$s = 0,\dots,|\mathcal{V}_i| $}
    \For{$r = s+1,\dots,|\mathcal{V}_i|$}
        \State Compute projection-based dissimilarity measure given in \eqref{d_sr}.
        \State Set $D^i_{s,r}=D^i_{r,s} = d(s,r)$.
    \EndFor
\EndFor
\EndFor
\vspace{0.5ex}
\Statex \textbf{Step 3: Mode clustering}
\State Apply agglomerative clustering with $M^i$ clusters to the distance matrix $D^i$.
\State Obtain mode labels $\ell_s \in \{1,\dots,M\}$ for each time segment $s = 0,\dots,\kappa $ from the chosen clustering method.
\vspace{0.5ex}
\State \Return $\{\ell_s\}_{s=0}^{\kappa}$.
\end{algorithmic}
\end{algorithm}

\begin{remark}
Here we choose the agglomerative clustering algorithm \citep{kaufman2009finding}  for time segment clustering, which requires the mode number in advance. We note that more sophisticated clustering  algorithms could be implemented without knowing the exact number of modes, which is a part of our future work.     
\end{remark}

\subsection{Clustering of Time segments into Modes}
After computing the pairwise dissimilarity $d(s,r)$ for all the time segments, we cluster the time segments into groups, where each group corresponds to a single mode. Intuitively, time segments that share the same underlying topology produce similar projection identities (Lemma~\ref{lem:mode_consistency}), resulting in small dissimilarities; time segments arising from different modes violate these identities, yielding large dissimilarities. Applying agglomerative clustering to the dissimilarity matrix therefore partitions the entire switching sequence into mode-consistent clusters.

Once clustering is complete, all the time segments assigned to the same mode are merged into a unified dataset. This aggregation compensates for the limited excitation available within each short time segment and provides a sufficiently rich dataset for reliable topology identification. Given the merged dataset for a mode, we then apply Algorithm~\ref{alg:algorithm2} to estimate the corresponding connection matrix (i.e., the graph associated with that mode). In this way, each cluster yields one reconstructed interaction topology.

\section{Simulation Results}
\label{Sec:simu}

To illustrate our theoretical findings, we present a numerical example of an OMAS modeled by \eqref{sys}, initially represented by a directed graph of eight agents , to which two additional agents join over time, as depicted in Figure \ref{fig:reference}. More precisely, we consider two modes for the eight-agent network, corresponding to modes 1-2 in Figure \ref{fig:directedweight}, and three modes for ten-agent network, corresponding to modes 3-5 in Figure \ref{fig:directedweight}. The switching sequence among these modes, during which agents may join or leave the system, is shown in Figure \ref{fig:switching}. The initial conditions of the agents are set randomly, and whenever new agents join, their initial states are also assigned randomly.  

\begin{figure}
    \centering
    \includegraphics[width=\linewidth]{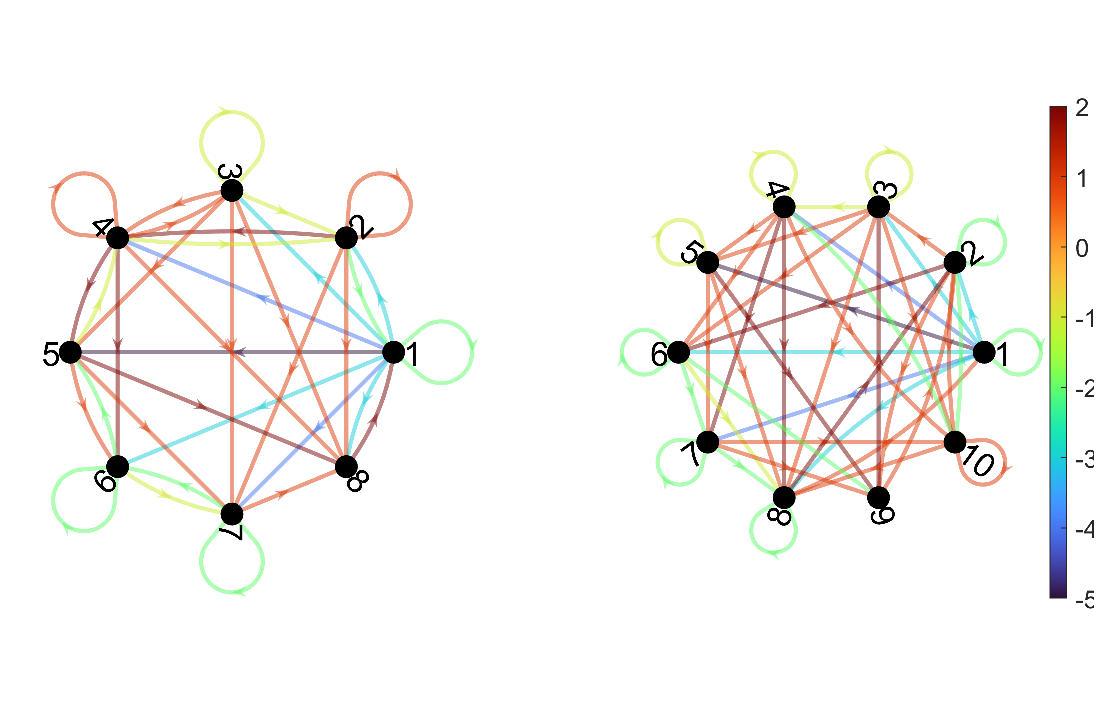}
   \vspace{-40pt} \caption{The reference graphs for $8$ agents and $10$ agents represented from $L_m$.}
    \label{fig:reference}
\end{figure}

\begin{figure}
    \centering
    \includegraphics[width=0.7\linewidth]{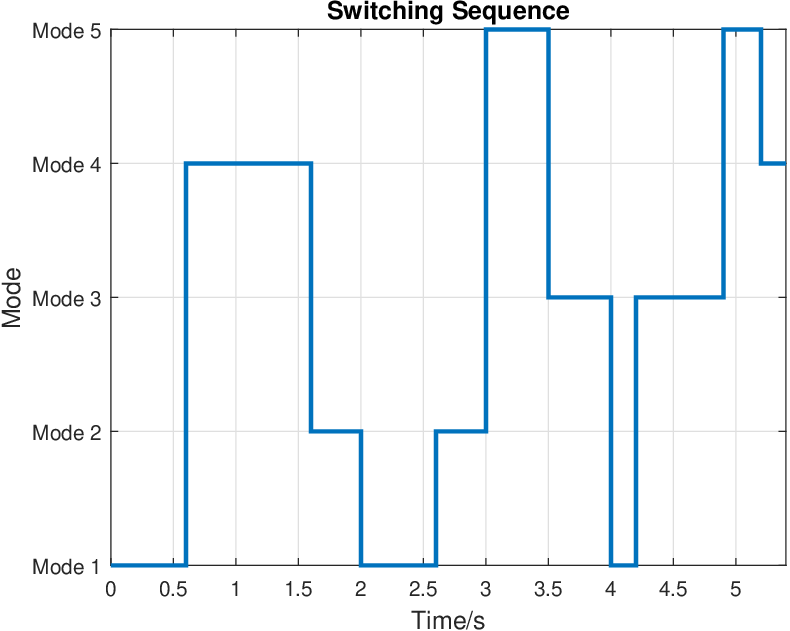}
    \caption{The evolution of the switching sequence  with time.}
    \label{fig:switching}
\end{figure}
The control input $u(t)$ is chosen as a combination of sinusoid signals with different frequencies. We choose different $L_m$ in Algorithm \ref{alg:algorithm1} for $8$-agent and $10$-agent networks and their reference graphs are displayed in Figure \ref{fig:reference}. The parameter $\tau$ in \eqref{eq:RK_system} is set to $\tau =0.05$. 

\begin{figure}
    \centering
    \includegraphics[width=\linewidth]{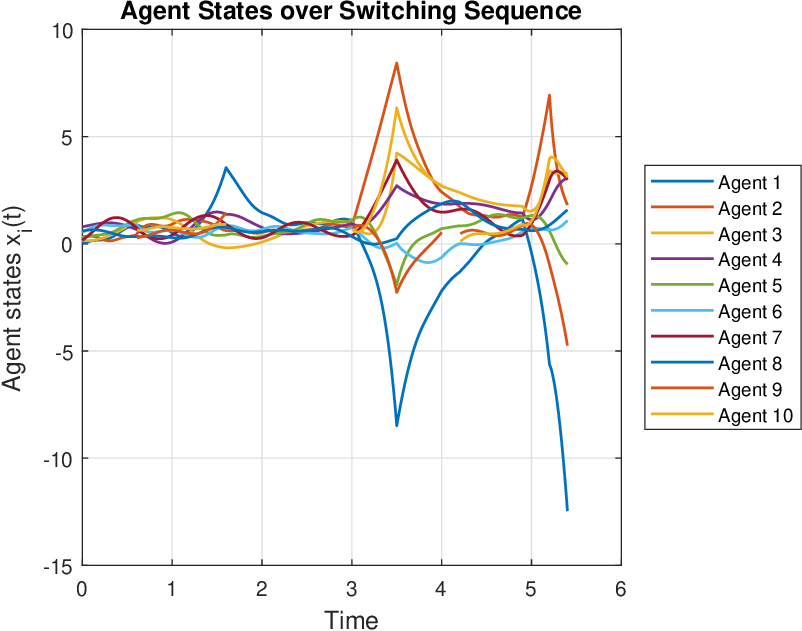}
    \caption{The evolution of the agent states  with time.}
    \label{fig:agent}
\end{figure}
\begin{figure}
    \centering
    \includegraphics[width=0.7\linewidth]{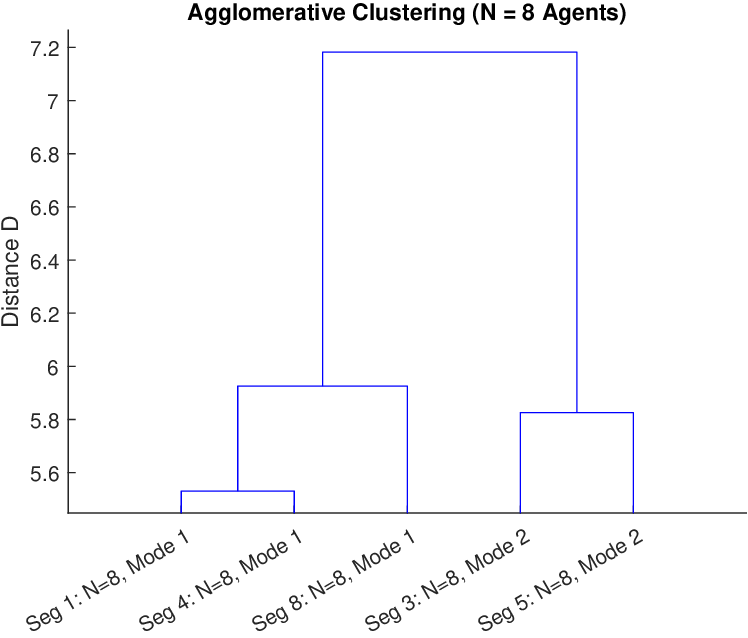}
    \caption{Dendrogram for 8 agents showing the projection-based dissimilarity measure among time segments. The number and mode number of time segments are the same with Fig.\ref{fig:switching}. There are $2$ modes and the clustering among time segments shows that the clustering of time segments are correct.} 
    \label{fig:distance_8}
\end{figure}
\begin{figure}
    \centering
    \includegraphics[width=0.7\linewidth]{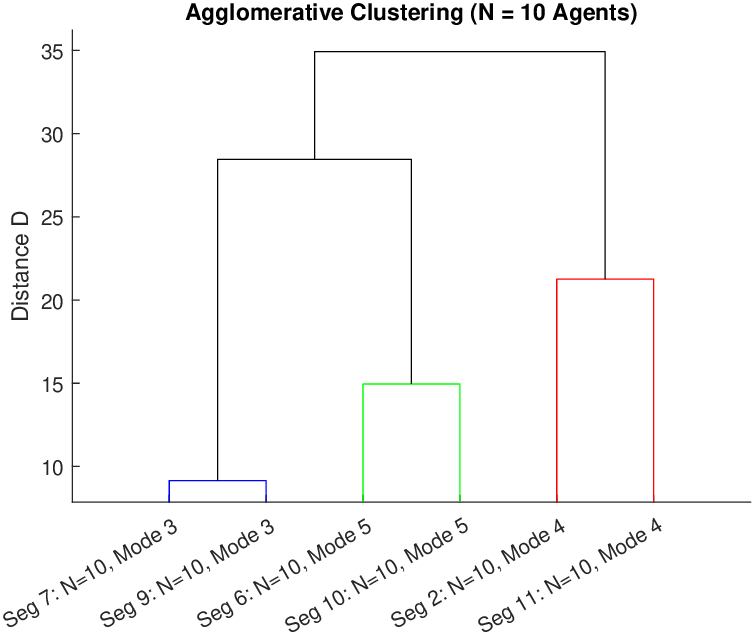}
    \caption{Dendrogram for 10 agents showing the projection-based dissimilarity measure time-segments. The number and mode number of  time segments are the same with Fig.\ref{fig:switching}. There are $3$ modes and the clustering among  time segments shows that the clustering of time segments are correct.}
    \label{fig:distance_10}
\end{figure}

\begin{table}[h] 
\centering
\caption{Estimation Errors Across Different Modes}
\label{tab:mode_errors} 
\begin{tabular}{lccccc}
\toprule
Modes number & 1 & 2 & 3 & 4 & 5 \\
 $\|\tilde L\|_F$ & 1.8e-11 & 2.1e-8 & 1.6e-8 & 1.0e-8 & 3.1e-7 \\
\bottomrule
\end{tabular}
\end{table}

The trajectories of states $x$ are displayed in Figure \ref{fig:agent}, where one can witness that $x$ are bounded. Firstly, we cluster time segments according to the size of the network using Algorithm \ref{alg:mode_clustering}. The dendrograms among  time segments for $8$ and $10$ agents are shown in Figure \ref{fig:distance_8} and \ref{fig:distance_10} and using agglomerative clustering, the clustering results for each time segment for modes are correct.
Then using  Algorithm \ref{alg:algorithm2}, we obtained the topology estimation errors for $5$ modes in Table \ref{tab:mode_errors},  which verifies Theorem~\ref{theore2}.

\section{Conclusion}
\label{Sec:con}
This paper presented a topology identification framework for OMAS characterized by dynamic node sets and fast switching interactions. By introducing a projection-based dissimilarity measure, we enabled reliable clustering of time segments and accurate mode-wise topology estimation without requiring long dwell time. 
Future research will focus on developing clustering algorithms that operate without prior knowledge of the number of modes, allowing data-driven mode discovery. Another important direction is to further relax the dwell time requirements, enabling identification under even more rapid and irregular switching conditions. A relevant numerical example shows the efficiency of our algorithms.

\bibliography{ref,references_pelin}

\end{document}